\pdfoutput=1
\documentclass{jocg}

\fancypagestyle{firststyle}
{
\fancyhead{}

}

\usepackage[internal]{paper}

\title{
    \MakeUppercase{Ptolemaic Indexing}
}
\author{Magnus Lie Hetland\/%
    \thanks{
        \affil{Dept.\ of Computer Science, Norwegian University of Science and
        Technology},
        \email{mlh@idi.ntnu.no}
    }
}

\begin{document}
\maketitle
\thispagestyle{firststyle}

\begin{abstract}
This paper discusses a new family of bounds for use in similarity search,
related to those used in metric indexing, but based on Ptolemy's inequality,
rather than the metric axioms. Ptolemy's inequality holds for the well-known
Euclidean distance, but is also shown here to hold for quadratic form metrics
in general, with Mahalanobis distance as an important special case.
The inequality is examined empirically on both synthetic and real-world data
sets and is also found to hold approximately, with a very low degree of error,
for important distances such as the angular pseudometric and several $L_p$
norms. Indexing experiments demonstrate a highly increased filtering power
compared to existing, triangular methods. It is also shown that combining the
Ptolemaic and triangular filtering can lead to better results than using
either approach on its own.
\end{abstract}

\section{Introduction}

In similarity search, data objects are retrieved based on their similarity to
a query object; as for other modes of information retrieval, the related
indexing methods seek to improve the efficiency the search. Two approaches
seem to dominate the field: spatial access
methods~\citep{Gaede:1998,Manolopoulos:2006,Samet:2006}, based on coordinate
geometry, and metric access
methods~\citep{Chavez:2001,Hjaltason:2003,Zezula:2006,Hetland:2009}, based on
the metric axioms. Similarity retrieval with spatial access methods is often
restricted to $L_p$ norms, or other norms with predictable behavior in
$\R^k\!$,
while metric access methods are designed to work with a broader class of
distances---basically any distance that satisfies the triangular inequality.
This gives the metric approach a wider field of application, by forgoing some
assumptions about the data, which in some cases results in lower performance.
Interestingly, even in cases where spatial access methods are applicable,
metric indexing may be superior in dealing with high-dimensional data, because
of its ability to bypass the so-called representational dimensionality and
deal with the intrinsic dimensionality of the data directly~\mbox{\citep[see,
e.g.,][]{Fu:2000}}.

It seems clear that there are advantages both to making strong assumptions
about the data, as in the spatial approach, and in working directly with the
distances, as in the metric approach. The direction taken in this paper is to
apply a new set of restrictions to the distance, separate from the metric
axioms, still without making any kind of coordinate-based assumptions. The
method that is introduced is shown to have a potential for highly increased
filtering power, and while it holds for the square root of any metric, it is
also shown analytically to apply to quadratic form metrics in general (with
Euclidean distance as an important special case), and empirically, as an
approximation with a very low degree of error, to important cases such as the
angular pseudometric, edit distance between strings, and several $L_p$ norms.

\bigskip

\noindent
This paper is a revised version of a preprint from~2009~\citep{Hetland:2009}.
While I have edited the paper for clarity, and at times adjusted the emphasis
on various topics, the substance remains the same. Since the original was
published, a few dozen papers have appeared that discuss or build on its
ideas, and I refer to a couple of the more relevant of these publications
throughout. The main contributions of the paper may be summed up as follows.

\begin{stmts}
\item It gives an example of indexing that goes beyond the established metric
    axioms, and demonstrates that doing so can have merit, poining the way to
    a broader field of exploration.
\item It introduces a filtering condition where the number of available bounds
    grows quadratically with the number of sample objects, and thus with the
    amount of space and the number of distance computations needed to perform
    the filtering. It also shows that this increase leads directly to an
    increase in filtering power. This introduces a tradeoff between the number
    of bounds and the number of false positives---a tradeoff that depends on
    the cost of the distance function.
\item It shows how to take the bounds beyond basic pivot filtering and apply
    them to more general index structures, potentially permitting more
    realistic, large-scale retrieval systems to be built.
\item It shows the equivalence of quadratic form metrics and Ptolemaic
    norm metrics (see Theorem~\ref{thm:qf}).
\end{stmts}

\noindent
The first point has already led others to look for axiom sets tailored to
specific data, for example~\citep[c.f.,][]{Bartos:2013,Bartos:2013b}. The
second point has led to indexing methods for the highly expensive
\emph{signature quadratic form distance}---methods that clearly beat the
competing metric indexes, as measured in actual search
time~\citep{Lokoc:2011}. The third point has been built upon to construct
fully formed, competetive index structures~\citep{Hetland:2013}. Finally, the
equivalence theorem has laid the foundations for the proof that the signature
quadratic form distance is a Ptolemaic metric~\citep{Lokoc:2011,Hetland:2013}.

This, then, is not primarily a paper on algorithm engineering; I do not
attempt to fully design specific index structures, with specific heuristics
for index traversal, or to perform exhaustive comparative
experiments; these are directions of research that have come to fruition
elsewhere.
Here, I simply aim to present the basic idea of Ptolemaic indexing, and
to demonstrate that it holds promise for wresting additional filtering power
from a given set of sample objects. And though I touch upon spaces that are
only partially Ptolemaic, in the interest of mapping out the territory, my
focus here is primarily on increased filtering power for exact search.

\section{Basic Concepts}\label{sec:basic}

The \emph{similarity} in similarity retrieval is usually formalized,
inversely, using a \defn{dissimilarity function}, a nonnegative real-valued
function $d(\cdot\,,\cdot)$ over some universe of objects, $\U$. For
typographic convenience, I will generally follow the
old-fashioned convention\footnote{See, for example, \citet{Wilson:1935}.}
of abbreviating $d(x,y)$ as $xy$.
A distance query consists of a query object $q$, and some form of distance
threshold, either given as a range (radius) $r$, or as a neighbor count $k$.
For the range query, all objects $o$ for which $qo\leq r$ are returned; for
the $k$-nearest-neighbor query ($k$NN), the $k$ nearest neighbors of $q$ are
returned, with ties broken arbitrarily. The $k$NN queries can be implemented
in terms of range queries in a quite general manner~\citep{Hjaltason:2003}.

For a function to qualify as a dissimilarity function (or \defn{premetric})
the value of $xx$ must be zero. It is generally assumed (mainly for
convenience) that the dissimilarity is \defn{symmetric} ($xy=yx$, giving
us a \defn{distance}) and \defn{isolating} ($xy=0\Leftrightarrow x=y$,
yielding a \defn{semimetric}). From a metric indexing perspective, it is
most crucial that the distance be
\defn{triangular}
(obeying the triangular inequality, $xz\leq xy+yz$). A \defn{metric} is
any symmetric, isolating, triangular dissimilarity function. Metrics are
related to the concept of \defn{norms}: A \defn{norm metric} is a metric
of the form $xy=\norm{x-y}$, where $x$ and $y$ are vectors and
\norm{\,\cdot\,} is a norm.

Many distances satisfy the metric properties, including several important
distances between sets, strings and vectors, and metric indexing is the main
approach in distance-based retrieval. The metricity, and in particular
triangularity, is exploited to construct lower bounds for efficient filtering
and partitioning. As discussed in an earlier tutorial~\citep{Hetland:2009},
there are several ways of using the triangular inequality in metric indexing,
usually by pre-computing the distances between certain sample objects (called
\defn{pivots} or \defn{centers}, depending on their use) and the other objects
in the data set. By computing the distances between the query and the same
sample objects (or some of them), triangularity can be used to filter out
objects that clearly do not satisfy the criteria of the query.

Even though most distance-based indexing has centered on the metric axioms,
this paper focuses on another property, known as \defn{Ptolemy's inequality},
which states that
\begin{equation}
xv\cdot yu \leq xy\cdot uv + xu\cdot yv\,,\label{eq:ptol}
\end{equation}
for all objects $x, y, u, v$. Premetrics satisfying this inequality are called
\defn{Ptolemaic}.\footnote{Note that Ptolemy's inequality neither implies
nor is implied by triangularity in general~\citep{Schoenberg:1940}.} In the
terms of the Euclidean plane: For any quadrilateral, the sum of the pairwise
products of opposing sides is greater than or equal to the product of the
diagonals (see \fig~\ref{fig:inequalities}).

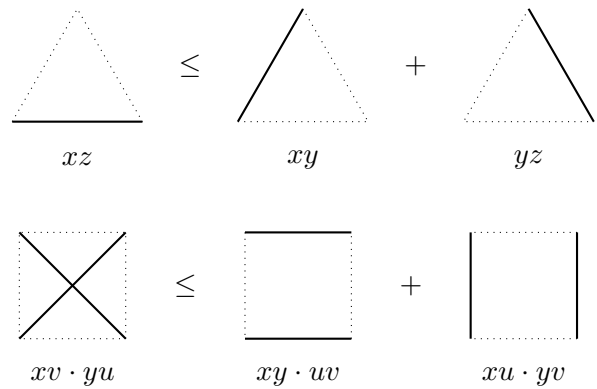
\begin{figure}[t]
\begin{center}
\subfigure{\begin{tikzpicture}
\begin{scope}
\path (-30:1) coordinate (z)
      (90:1) coordinate (y)
      (210:1) coordinate (x);
      \draw[thick] (x) -- (z);
\draw[thin,dotted] (x) -- (y) -- (z);
\path ($(x)!.5!(z)$) coordinate (xz);
\draw (xz) ++(0,-.5) node {$xz$};
\draw ($(xz)!.5!(y)$) ++(1.5,0) node {$\leq$};
\end{scope}
\begin{scope}[shift={(3,0)}]
\path (-30:1) coordinate (z)
      (90:1) coordinate (y)
      (210:1) coordinate (x);
\draw[thin,dotted] (x) -- (z) -- (y);
\draw[thick] (x) -- (y);
\path ($(x)!.5!(z)$) coordinate (xz);
\draw (xz) ++(0,-.5) node {$xy$};
\draw ($(xz)!.5!(y)$) ++(1.5,0) node {$+$};
\end{scope}
\begin{scope}[shift={(6,0)}]
\path (-30:1) coordinate (z)
      (90:1) coordinate (y)
      (210:1) coordinate (x);
\draw[thin,dotted] (y) -- (x) -- (z);
\draw[thick] (y) -- (z);
\draw ($(x)!.5!(z)$) ++(0,-.5) node {$yz$};
\end{scope}
\end{tikzpicture}}\\[5mm]
\subfigure{\begin{tikzpicture}
\begin{scope}
\path (-45:1) coordinate (a)
      (45:1) coordinate (b)
      (135:1) coordinate (c)
      (225:1) coordinate (d);
\draw[thick] (a) -- (c) (b) -- (d);
\draw[thin,dotted] (a) -- (b) -- (c) -- (d) -- cycle;
\draw ($(a)!.5!(d)$) ++(0,-.5) node {$xv\cdot yu$};
\end{scope}
\draw (1.5,0) node {$\leq$};
\begin{scope}[shift={(3,0)}]
\path (-45:1) coordinate (a)
      (45:1) coordinate (b)
      (135:1) coordinate (c)
      (225:1) coordinate (d);
\draw[thin,dotted]
    (a) -- (b) (c) -- (d);
\draw[thick] (a) -- (d) (b) -- (c);
\draw ($(a)!.5!(d)$) ++(0,-.5) node {$xy\cdot uv$};
\end{scope}
\draw (4.5,0) node {$+$};
\begin{scope}[shift={(6,0)}]
\path (-45:1) coordinate (a)
      (45:1) coordinate (b)
      (135:1) coordinate (c)
      (225:1) coordinate (d);
\draw[thin,dotted]
    (a) -- (d) (b) -- (c);
\draw[thick] (a) -- (b) (c) -- (d);
\draw ($(a)!.5!(d)$) ++(0,-.5) node {$xu\cdot yv$};
\end{scope}
\end{tikzpicture}}
\caption{An illustration of the triangle inequality (top) and Ptolemy's
inequality (bottom).}\label{fig:inequalities}
\end{center}
\end{figure}

It is a well-known fact that Euclidean distance is
Ptolemaic~\citep{Deza:2006}. There are many proofs for this, some quite
involved, but it can be shown quite simply,
using the idea of \emph{inversion}~\citep{Smith:1994}. A sketch of
the argument is given in \fig~\ref{fig:inv}.
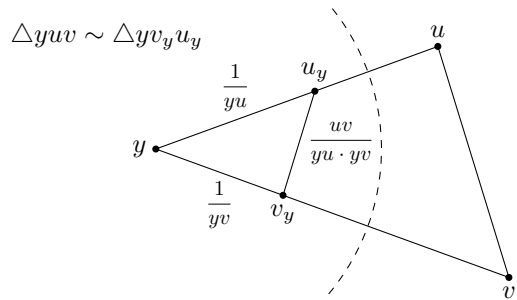
\begin{figure}
\begin{center}
\begin{tikzpicture}[scale=1.5]
\small
\def\r{0.75pt}
\path (0,0)         coordinate (y)
      (20:1.5cm)    coordinate (u_y)
      (20:2.66cm)   coordinate (u)
      (-20:1.2cm)   coordinate (v_y)
      (-20:3.33cm)  coordinate (v)
      ;
\draw[dashed] (y) + (-40:2cm) arc (-40:40:2cm);
\filldraw \foreach \p in {u,u_y,y,v,v_y} {
    (\p) circle (\r)
};
\draw (y)   node[left]  {$y$}
      (u_y) node[above] {$u_y$}
      (u)   node[above] {$u$}
      (v_y) node[below] {$v_y$}
      (v)   node[below] {$v$};
\draw (u_y) -- (u) -- (v) -- (v_y);
\def\s{0.8}
\draw (y) -- (u_y)
      node[midway,above] {\scalebox{\s}
      {$\displaystyle\frac{1}{yu}$}}
      -- (v_y)
      node[midway,right] {\scalebox{\s}
      {$\displaystyle\frac{uv}{yu\cdot yv}$}}
      -- (y)
      node[midway,below] {\scalebox{\s}
      {$\displaystyle\frac{1}{yv}$}}
      ;
\draw (.5cm,1cm) node[left] {$\bigtriangleup yuv\sim \bigtriangleup yv_yu_y$};
\end{tikzpicture}
\caption{Getting Ptolemy's inequality from the triangular inequality through
inversion in Euclidean space: The points are inverted with respect to $y$,
with an inversion radius of unity (dashed). Adding a third point, $x$, and
applying the triangular inequality to $u_y$, $v_y$ and $x_y$ yields Ptolemy's
inequality for $x$, $y$, $u$ and $v$.}\label{fig:inv}
\end{center}
\end{figure}
What is perhaps
less well known is that there is a natural connection between Ptolemaic
metrics on vector spaces and a generalization of Euclidean distance---a family
of distances collectively referred to as \defn{quadratic form distance}. A
quadratic form distance may be expressed as

\[
d(\vx,\vy) = \sqrt{\textstyle{\sum_{i=1}^n\sum_{j=1}^n
                a_{ij}(x_i-y_i)(x_j-y_j)}}\,,
\]
or, in matrix notation, $\sqrt{\vz'\!\A\vz}$, where $\vx$ and $\vy$ are
vectors, and $\vz=\vx-\vy$. The weight matrix $\A=[a_{ij}]$ is a measure of
``unrelatedness'' between the dimensions, which uniquely defines the distance.
We can, without loss of generality, assume that \A is symmetric, as any
antisymmetries will have no bearing on the distance~\citep{Hafner:1995}. In
order for the distance to be a metric, \A must also be
positive-definite.\footnote{Some sources require only
positive-semidefiniteness~\citep[e.g.,][]{Zezula:2006}, but this would result
in a pseudometric, allowing a distance of zero between different objects. It
is possible to relax the requirement on \A by adding requirements to the
inputs~\citep{Hafner:1995}.}

Quadratic form distances take into account possible correlations between the
dimensions of the vector space, which makes them especially suited for
comparing histograms~\citep[see, e.g.,][]{Bernas:2008}. For example,
when comparing color histograms, it might be natural to compare the bin (that
is, dimension) for orange to those of red and yellow~\citep{Hafner:1995}. If
\A is restricted to a diagonal matrix, a weighted Euclidean distance results,
with the identity matrix leading to ordinary Euclidean distance (see
\fig~\ref{fig:balls}). An important kind of quadratic form distance is the
Mahalanobis distance, where \A is normally set to the inverse of the
covariance matrix of a data sample.

The fact that this important family of distances is usable with the techniques
presented in this paper is expressed in the following theorem.

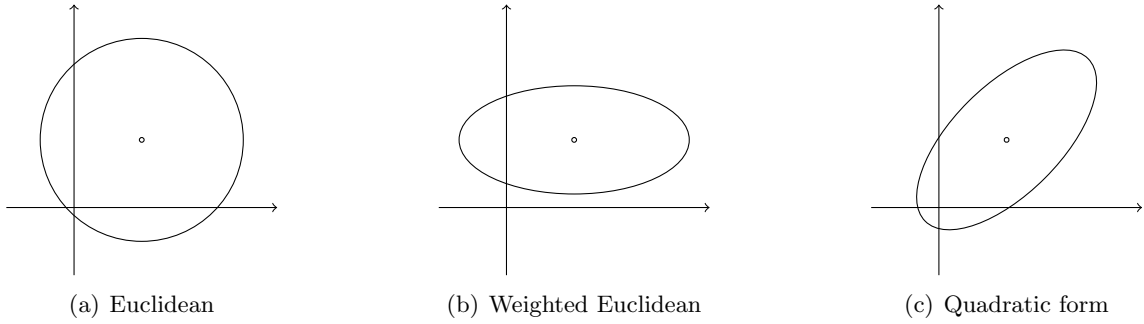
\begin{figure}
\def\figscale{.9}
\subfigure[Euclidean]{\begin{tikzpicture}[scale=\figscale]
\draw[->] (0,-1) -- (0,3);
\draw[->] (-1,0) -- (3,0);
\draw (1,1) circle (1.5) circle (1pt);
\end{tikzpicture}}
\hfill
\subfigure[Weighted Euclidean]{\begin{tikzpicture}[scale=\figscale]
\draw[->] (0,-1) -- (0,3);
\draw[->] (-1,0) -- (3,0);
\draw (1,1) ellipse (1.7 and 0.8) circle (1pt);
\end{tikzpicture}}
\hfill
\subfigure[Quadratic form]{\begin{tikzpicture}[scale=\figscale]
\draw[->] (0,-1) -- (0,3);
\draw[->] (-1,0) -- (3,0);
\draw (1,1) [rotate around={45:(1,1)}] ellipse (1.7 and 0.8) circle (1pt);
\end{tikzpicture}}
\caption{Metric balls for three related distance types in
$\R^2$.}\label{fig:balls}
\end{figure}

\begin{theorem}\label{thm:qf}
A distance function on $\R^n$ is a quadratic form metric if and only if it is
a Ptolemaic norm metric.
\end{theorem}
\begin{proof}
Let the quadratic form metric $d(\vx,\vy)$ be $\sqrt{\vz'\!\A\vz}$, where
$\vz=\vx-\vy$. Because \A is symmetric positive-definite, $\vz'\!\A\vz$
defines an inner product, making any such $d$ a norm metric based on an inner
product norm (a norm of the form $\norm{\vz}=\sqrt{\inner{\vz}{\vz}}$, where
\inner{\,\cdot\,}{\cdot\,} is an inner product). Note that the converse also
holds: Any inner product norm distance on $\R^n$ can be expressed as a
quadratic form metric (with a positive-definite \A).
It is known that a norm metric is Ptolemaic if and only if its norm is an
inner product norm~\citep{Schoenberg:1952,Deza:2006}, which gives us the
theorem.
\end{proof}

\noindent
As pointed out by \citet{Skopal:2011}, if $\A$ is static, a
distance-preserving transform into Euclidean space can be used to eliminate
the quadratic form from the search. This is not possible, however, if $\A$ is
constructed on demand, as in the \emph{signature} quadratic form distance, for
example~\citep{Beecks:2010}. Computing such distances will tend to be
expensive, making them prime candidates for Ptolemaic indexing.

Beyond quadratic forms, there are certainly many other Ptolemaic metrics (with
the discrete metric, $d(x,y)=1\Leftrightarrow x\neq y$, as an obvious
example, and the chordal metric on the unit Riemann sphere as a possibly
less obvious one~\citep{Malesevic:2006}). In fact, for any metric $d$, the
metric $\sqrt{d(\,\cdot\,)}$ is Ptolemaic~\citep{Foertsch:2006}. In terms of distance
orderings, and therefore similarity queries, this new metric is equivalent to
the original, meaning that the techniques in this paper are applicable to all
metrics. However, the transform will
increase its intrinsic dimensionality~\citep{Skopal:2007}, generally making
the process of indexing harder. As shown in \sect~\ref{sec:approx}, several
non-Ptolemaic metrics may be ``sufficiently Ptolemaic'' to be used with
the Ptolemaic indexing techniques \emph{without} such a transform, although
only for approximate queries.

\section{Related Work}

As discussed at length
elsewhere~\citep{Chavez:2001,Hjaltason:2003,Samet:2006,Zezula:2006,Hetland:2009},
there are many published methods that deal with indexing distances based on
the metric properties. While there has been an increasing focus on reducing
\textsc{i/o} and general \textsc{cpu} time, the primary aim of most
publications has been minimizing the number of distance computations, based on
the assumption that the distance is highly expensive to compute.\footnote{This
assumption may very well stem from the seminal work of \citet{Feustel:1982},
where calculation of the metric involved comparing every permutation of the
node sets of two graphs.} While focusing exclusively on this one performance
criterion may not be altogether realistic, yielding methods with
linear query time~\citep{Mico:1994,Pedreira:2007} or quadratic memory
use~\citep{Vidal:1986,Figueroa:2006}, for example, it has proven a useful
foundation on which methods with more nuanced performance properties could be
built~\citep[e.g.,][]{Zezula:1996,Dohnal:2004,Fredriksson:2006,Brisaboa:2008}.
This paper also focuses on minimizing distance computations, setting aside
related questions of algorithm engineering for later.

The main mechanism through which distance calculations may be avoided is
through various forms of \emph{filtering} or \emph{exclusion}, using lower
bounds~\citep{Hetland:2009}.\footnote{The converse, \emph{inclusion} using
upper bounds, is also possible, but less frequently useful, simply because
most of the data should normally be excluded from the result set.} If the
structure of the data objects is known, then very precise, yet cheap, lower bounds
can be constructed~\citep{Hetland:2004}, but for metric indexing, only the
general properties of the distance may be used. As described in
\sect~\ref{sec:basic}, this is done by storing precomputed distances, or
ranges of distances, involving the data set and certain sample objects
(\emph{centers} or \emph{pivots}).\footnote{In the case of generalized
hyperplane indexing, the information stored is simply which center is closest
to a given object, and this is implicit in the structure.}

Rather than listing existing indexing methods, only the most relevant of the
basic principles will be addressed here. Two rather general theorems contain
the majority of the indexing principles as special cases. The first of these,
dealing with metric balls, is given here without proof. The second, dealing
with generalized hyperplanes, as well as more details and proofs for both
theorems can be found in the aforementioned tutorial~\citep{Hetland:2009}, the
survey by \citet{Hjaltason:2003} or the textbook by \citet{Zezula:2006}. In
the following, let $o$, $p$ and $q$ be objects in a universe \U, and let the
implicit distance $d$ be a metric over \U.

\begin{theorem}\label{thm:ball}
The value of $qo$ may be bounded as
\[
\max\{
\lo{op}-\hi{qp},
\lo{qp}-\hi{op}
\}\leq qo \leq \hi{qp} + \hi{op}\,,
\]
where $\lo{uv} \leq uv \leq \hi{uv}$, for any objects $u$, $v$ in \U. \qed
\end{theorem}

\noindent
The expressions $\lo{uv}$ and $\hi{uv}$ refer to known lower and upper bounds
for $uv$,
as, in some cases, the exact distances may not be known. For example, $p$
may be the center of a metric ball containing $o$, with covering radius $r$.
In that case $\lo{po}=0$ and $\hi{po}=r$. If the query--pivot distance is
known, we get the lower bound
\begin{equation}\label{eq:balllb}
qp-r \leq qo\,,
\end{equation}
which is exactly the bound used to check for overlap between a query ball and
a bounding ball in a metric tree, for example.

\section{Ptolemaic Pivot Filtering}\label{sec:pivots}

In a manner similar to Theorem~\ref{thm:ball}, Ptolemy's inequality may be
used to construct lower bounds for filtering. In the following, the technique
known as \defn{pivoting} is used. The derivation of a more general bound is
deferred to \sect~\ref{sec:general}.
Triangular pivoting is based on the following lower bound, a special case of
the one in Theorem~\ref{thm:ball}, where the distances are known exactly:
\begin{equation}
qo \geq \vert qp - op\vert\label{eq:piv}
\end{equation}
Here, $q$ is the query object, $o$ is a candidate result object, while $p$ is
a so-called \defn{pivot} object, whose function is to help construct the
bound.
This bound is sometimes known as the \emph{inverse} triangular inequality, and
follows from basic restructuring of the original:
\begin{align}
op + qo &\geq qp\nonumber\\
     qo &\geq qp - op\label{eq:first}\,,\\
\intertext{and, in the same manner,}
     qo &\geq op - qp\label{eq:second}\,.
\end{align}
Together, (\ref{eq:first}) and (\ref{eq:second}) lead directly to
(\ref{eq:piv}).
The bound is normally strengthened by using a set of several pivots, $P$:
\begin{equation*}
qo \geq \max_{p\,\in\,P}\, |qp - op|
\end{equation*}
A similar derivation can be made for Ptolemaic distances. In the following,
$q$, $p$, and $o$ retain their previous meaning, but we also add another pivot
object, $s$:
\begin{align}
qs\cdot op + qo\cdot ps &\geq qp\cdot os\nonumber\\
qo\cdot ps &\geq qp\cdot os - qs\cdot op\nonumber\\
qo &\geq (qp\cdot os - qs\cdot op)/ps\label{eq:ptfirst}
\end{align}
Here we can maximize over all pairs of pivots:
\begin{equation}
qo \geq \max_{p,s\,\in\,P}\,
\frac{qp\cdot os - qs\cdot op}{ps}\label{eq:pt}
\end{equation}
By exchanging $p$ and $s$ in (\ref{eq:ptfirst}), we could exchange the terms
in the numerator, allowing us to use the absolute value in (\ref{eq:pt}). This
would not strengthen the bound as it stands, but would allow us halve the
number of pivot pairs examined.

\noindent
For the normal pivoting bound to be useful, the pivot should be closer to
either the query or to the candidate object; the difference in the two
distances is what gives the bound its filtering power. For the Ptolemaic
pivoting bound, it seems that one way of getting good results would be to have
one pivot close to the query, while the other is close to the candidate
object, giving a high value for the numerator in (\ref{eq:pt}). However, this
intuition is tempered by the denominator, which dictates that the pivots
should also be close to each other. Invariably, the tradeoff here will need to
be based on empirical considerations.
As long as the distance matrix between the pivots is precomputed, the bound
can be computed for every pair of pivots, and the maximum used as the final
bound. As will be shown in \sect~\ref{sec:results}, this Ptolemaic pivoting
bound is a significant improvement over the classical triangular one.

The difference between triangular and Ptolemaic pivoting in the Euclidean
plane is illustrated in \fig~\ref{fig:heatmap}, where the ratio between bound
and distance from a query at $(-1,0)$, with the pivots placed at $(0,0)$ and
$(1,0)$, is plotted for each point, where zero (a non-informative lower bound)
is black and one (a perfect bound) is white.
While the the bound in~(\ref{eq:pt}) is the one examined further in this
paper, the ideas of Ptolemaic indexing have wider implications. In the
following section,
a more general theorem (Theorem~\ref{thm:ptball}) is given, which is a
Ptolemaic analogue of Theorem~\ref{thm:ball}.

\begin{figure}
\begin{center}
\subfigure[Triangular pivoting]{\begin{tikzpicture}
    \begin{axis}[enlargelimits=false,axis on top,small,
        scale only axis,
        height=5cm, width=5cm,
        xtick={0}, ytick={0},
        xticklabels={},yticklabels={},
        cycle list name=black white,
        ]%
        \addplot graphics
        [xmin=-3,xmax=3,ymin=-3,ymax=3]
        {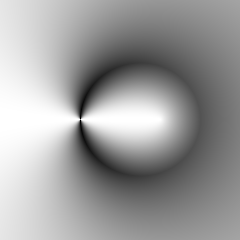};
        \addplot+[only marks,mark=asterisk]%
        coordinates
        {(0,0) (1,0)};
        \addplot+[only marks,mark=*]%
        coordinates
        {(-1,0)};
    \end{axis}
\end{tikzpicture}}
\hspace{.12\linewidth}
\subfigure[Ptolemaic pivoting]{\begin{tikzpicture}
    \begin{axis}[enlargelimits=false,axis on top,small,
        scale only axis,
        height=5cm, width=5cm,
        xtick={0}, ytick={0},
        xticklabels={},yticklabels={},
        cycle list name=black white
        ]
        \addplot graphics
        [xmin=-3,xmax=3,ymin=-3,ymax=3]
        {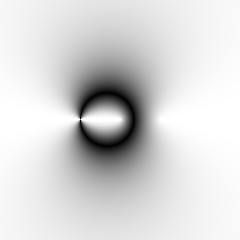};
        \addplot+[only marks,mark=asterisk]
        coordinates
        {(0,0) (1,0)};
        \addplot+[only marks,mark=*]%
        coordinates
        {(-1,0)};
    \end{axis}
\end{tikzpicture}}
\caption{Accuracy, measured as the ratio between the pivoting bound and the
    actual distance, from 0\%
The three points represent, from left to right, a query and two
pivots.}\label{fig:heatmap}
\end{center}
\end{figure}
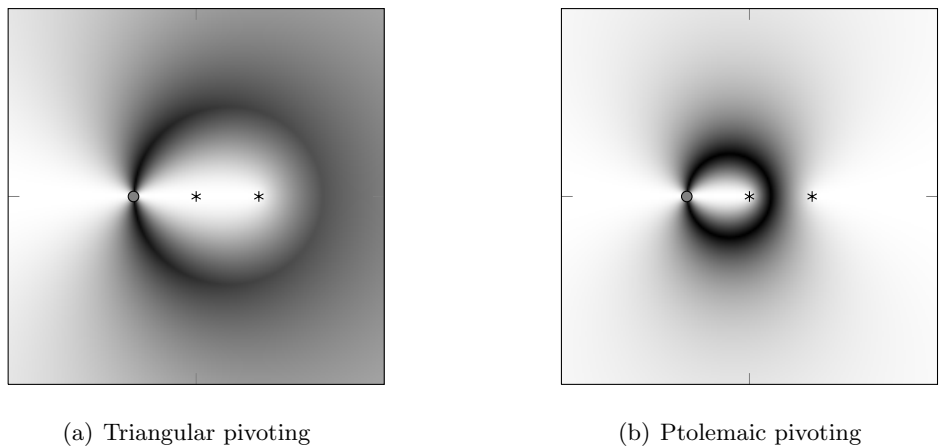

\section{Generalizing the Ptolemaic Bound}\label{sec:general}

The following theorem generalizes the bound~(\ref{eq:pt}), as a Ptolemaic
analogue of Theorem~\ref{thm:ball}. This generalization is included as a
starting point for new indexing methods, and is not examined empirically in
this paper.

\begin{theorem}\label{thm:ptball}
Let $o$, $p$, $q$ and $s$ be objects in a universe \U, and let $d$ be a
Ptolemaic distance over \U. The value of $qo$ may then be bounded as
\[
\frac{1}{\hi{ps}}\cdot
\max\left\{\begin{array}{l}
\lo{qs}\cdot\lo{op}-\hi{qp}\cdot\hi{os},\\
\lo{qp}\cdot\lo{os}-\hi{qs}\cdot\hi{op}
\end{array}\right\}
\leq qo
\leq \frac{1}{\lo{ps}}\cdot (\hi{qp}\cdot\hi{os} + \hi{qs}\cdot\hi{op})\,,
\]
where $\lo{uv} \leq uv \leq \hi{uv}$, for any objects $u$, $v$ in
\U.%
\end{theorem}
\begin{proof}

The two cases of the lower bound correspond to the two possible orderings of
the products in the numerator of the lower bound (\ref{eq:pt}), both of which
are permissible. The only change here is that instead of the exact values, we
use upper and lower limits. The lower limits occur before the subtractions,
while the upper limits occur after, as well as in the denominator. Given that
all distances are non-negative, the lower (resp.\@ upper) bounds also provide
lower (resp.\@ upper) bounds on the products. Therefore, these substitutions can
only lower the value of the bound, and hence the inequality still holds. The
upper bound follows from the Ptolemaic inequality $$ps\cdot qo\leq qp\cdot
os+qs\cdot op\,.$$ Dividing by $ps$ and safely substituting upper limits in
the numerator and a lower limit in the denominator, we arrive at the upper
bound.
\end{proof}

\noindent
The applications of this theorem to pivot filtering have already been
discussed in \sect~\ref{sec:pivots}. However, its metric analogue
Theorem~\ref{thm:ball} is also used for overlap checking with balls and
shells, which is what the upper and lower limits to the distances represent.
The notion of overlapping metric balls is inherently triangular, and does not
directly translate to Ptolemaic distances. We are still able to exploit
similar information, but we are in the somewhat unusual situation of working
with two balls (or shells) at once. Take the upper bound, in the case where we
know $qp$ and $qs$ (and, of course, $ps$). In order to apply this bound (for
automatic inclusion of $o$), we would have to know that $o$ falls inside two
balls, one around $p$ and one around $s$, with radii $\hi{op}$ and $\hi{os}$,
respectively.
While this sort of ``double containment'' is not the norm in current metric
indexing methods, it is certainly not impossible to implement. One would
simply need to let each region be represented by two distance balls, and
maintain two covering radii---one for each center. There is an inherent
tradeoff here: If the objects are far apart, the covering radii will
necessarily become quite large; however, if we move them closer, the upper
bound will increase.

In considering the \emph{lower} bound, we see something interesting: If we
envision a structure with covering radii around both $p$ and $s$, the lower
limits $\lo{op}$ and $\lo{os}$ both become zero, leaving us also with a total
lower bound of zero. We see that, as for the pivot filtering case, we may need
for one of the pivots to be more ``query-like,'' and the pivots need to be
different from each other. For example, if the query is close to $p$ but far
from $s$, and the converse holds for the object, the lower bound will be high.

However, this may not be enough. Consider the case where the object falls
within a ball with radius $r$ around $p$ (giving us $\lo{os}=ps-r$); we then
get the following variation of the lower bound from (\ref{eq:balllb}):

\[
qo \geq \frac{1}{ps}\cdot (qp\cdot (ps - r) - r\cdot qs) = qp -
r\cdot\frac{qp+qs}{ps}
\]
The only difference from the triangular condition is the scaling factor
$(qp+qs)/ps$, which we can see as regulating the influence of the radius. If
the query lies directly between $p$ and $s$ (that is, $ps=qp+qs$) this new
bound is, in fact, exactly equivalent to the triangular one. However, for all
other cases, the new bound is \emph{worse}.

What's missing is the ``other half,'' as it were: an \emph{inverted} ball
around $s$, excluding $o$, giving us a proper \lo{os}, in addition to the
covering radius, \hi{op}. This would be available in a situation where we have
covering shells from each pivot to its sibling regions (as in
\textsc{gnat}~\citep{Brin:1995} and its descendants, or the PM-Tree, which
uses subtree shells around global pivots~\citep{Skopal:2004a}), or where we
use ``inside/outside'' partitioning with multiple pivots simultaneously (as in
D-Index~\citep{Dohnal:2003}). In such cases, where both upper and lower bounds
are available for both $op$ and $os$, the bound in Theorem~\ref{thm:ptball}
can be used directly, substituting exact values for $ps$, $qs$
and~$qp$.\footnote{This, in fact, is the technique used in implementing the
Ptolemaic PM-Tree~\citep{Hetland:2013}.} See \fig~\ref{fig:shells} for an
example.

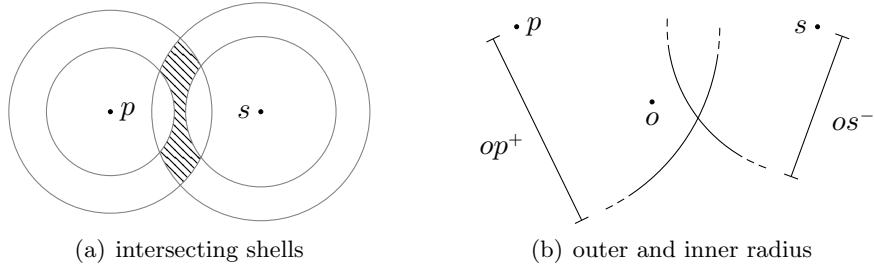
\begin{figure}
\begin{center}
\subfigure[intersecting shells]{
\begin{tikzpicture}[scale=.5]
\def\r{1.5pt}
\path (-2,0)         coordinate (p)
      (+2,0)         coordinate (s)
      ;

\begin{scope}
    \path[clip]
        (p) circle (2.7)
    ;
    \path[clip]
        (s) circle (2.9)
    ;
\fill[pattern=north west lines] (-10,-10) rectangle (10,10);
\fill[white]
    (p) circle (1.7)
    (s) circle (2);
\end{scope}

 \draw (p)   node[right]  {$p$}
       (s)   node[left] {$s$}
       ;
\filldraw \foreach \p in {p,s} {
    (\p) circle (\r)
};

\draw[gray]
    (p) circle (2.7)
    (p) circle (1.7)
    (s) circle (2)
    (s) circle (2.9)
    ;
\end{tikzpicture}
}
\qquad
\subfigure[outer and inner radius]{
\begin{tikzpicture}
\def\r{0.75pt}
\path (-2,0)         coordinate (p)
      (+2,0)         coordinate (s)
      (-.2,-1)         coordinate (o)
      ;
\filldraw \foreach \p in {p,s,o} {
    (\p) circle (\r)
};
 \draw (p)   node[right]  {$p$}
       (s)   node[left] {$s$}
       (o)   node[below] {$o$}
       ;
\draw
    (p) + (-55:2.7) arc (-55:-10:2.7)
    ;
\draw[densely dashed]
    (p) + (-64:2.7) arc (-64:-55:2.7)
    (p) + (-10:2.7) arc (-10:0:2.7);
    ;
\draw
    (s) + (240:2) arc (240:190:2);
\draw[densely dashed]
    (s) + (250:2) arc (250:240:2)
    (s) + (190:2) arc (190:176:2)
    ;

\draw[|-|] (p) ++(-154:10pt) -- +(-64:2.7);
\path (p) ++(-154:25pt) -- +(-64:1.35)
    node {$\hi{op}$};
\draw[|-|] (s) ++(340:10pt) -- +(250:2);
\path (s) ++(340:25pt) -- +(250:1)
    node {$\lo{os}$};
    ;
\end{tikzpicture}
}
\end{center}
\caption{Two intersecting shells of a PM-Tree or some similar structure. Both
shells contain the objects of the region, se we are free to use the upper
radius of one and the lower of the other for added Ptolemaic
filtering.}\label{fig:shells}
\end{figure}

Finally, for Ptolemaic \emph{metrics}, it is possible to combine the upper or
lower bounds generated by Ptolemaic pivoting with various metric regions.
This is similar to the technique used in such metric index structures as the
CM-Tree~\citep{Aronovich:2007} and \textsc{tlaesa}
\citep{Mico:1996,Tokoro:2006}, where a lower bound to the center object is
used in ball overlap checking, and in Hybrid Dynamic
SA-Tree~\citep{Arroyuelo:2003}, where the same approach is taken in hyperplane
filtering.
The idea is to be able to filter out an entire region without computing the
distance to its representative objects (such as centers or vantage points).
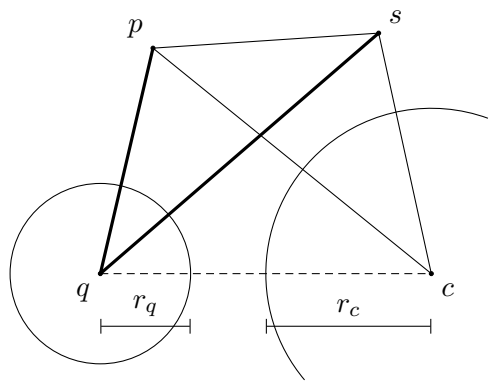
\begin{figure}
\begin{center}
\begin{tikzpicture}
\def\r{0.75pt}
\path (-1.5,+1.0) coordinate (p)
      (+1.5,+1.2) coordinate (s)
      (-2.2,-2) coordinate (q)
      (+2.2,-2) coordinate (c)
      ;
\filldraw \foreach \p in {p,s,q,c} {
    (\p) circle (\r)
};
\draw
    (p) node[above left] {$p$}
    (s) node[above right] {$s$}
    (q) node[below left] {$q$}
    (c) node[below right] {$c$}
    ;
\draw[very thick] (q) -- (p) (q) -- (s);
\draw
    (p) -- (s) -- (c)
    (p) -- (c)
    ;
\draw
    (q) circle (1.2)
    ;
\draw
    (c) +(220:2.2) arc (220:70:2.2)
    ;
\draw[|-|]
    (q) ++(0,-.7) -- +(1.2,0)
    node[midway, above] {$r_q$}
    ;
\draw[|-|]
    (c) ++(0,-.7) -- +(-2.2,0)
    node[midway, above] {$r_c$}
    ;
\draw[densely dashed] (q) -- (c);
\end{tikzpicture}
\end{center}
\caption{Computing pivot distances (heavy) at the start of the search
potentially lets us eliminate regions without examining their
centers. The available distances (solid) combine to form a lower bound on $qc$;
if this bound is greater than $r_q + r_c$, overlap is impossible.}
\label{fig:doublebound}
\end{figure}
For example, for a ball with center $c$ and covering radius $r_c$, as well as
a query radius $r_q$ the normal discard criterion would be $r_q < qc-r_c$.
Instead, with the pivots $p$ and $s$, we could replace $qc$ with
a Ptolemaic lower bound (see \fig~\ref{fig:doublebound}). The criterion then
becomes
$$
r_q < \frac{qp\cdot cs - qs\cdot cp}{ps} - r_c\,.
$$
For a hyperplane defined by the two centers $u$ and $v$, the region
corresponding to $v$ can normally be excluded if $2r_q < qv-qu$. Again, we can
avoid calculating the distance $qv$ by using a lower bound in its place, with
the criterion
$$
2r_q < \frac{qp\cdot vs - qs\cdot vp}{ps} - qu\,.
$$
In both cases, multiple pivots would, of course, strengthen the bounds.
Similar techniques could also be used with metric shells, or even the vantage
point regions of the VP-Tree~\citep{Uhlmann:1991,Yianilos:1993} and the
D-Index~\citep{Dohnal:2003}.

\section{Experimental Results}\label{sec:results}

Two sets of experiments have been performed to explore the potential
usefulness of Ptolemy's inequality in distance indexing: The first set of
experiments evaluate its filtering power (using pivot filtering), while the
second evaluates its recall rates for non-Ptolemaic metrics, for possible use
in approximate search.

\subsection{Filtering Power for Ptolemaic Metrics}\label{sec:filtering}

The first set of experiments were designed to compare the filtering power of
the Ptolemaic and triangular approaches. \fig~\ref{fig:filtering} shows the
results for pivot filtering, with both the triangular and Ptolemaic lower
bounds.

The first data set consisted of uniformly random 10-dimensional
vectors, and the queries used for evaluation were drawn from the same
distribution.
For the second set, 20-dimensional data and query vectors were drawn in equal
proportion from ten gaussian clusters (with $\sigma^2=0.1$), centered around
vectors drawn uniformly random from the unit hypercube. This clustering
approach is similar to that used by \citet{Zezula:1996}.
The last data set consisted of
64-dimensional image histograms under a quadratic form distance
similar to
that described by \citet{Hafner:1995}, using Euclidean distance in
\textsc{rgb} space as the basis for the weight matrix.\footnote{The matrix
used was $a_{ij}=1-d_{ij}/d_{\textrm{max}}$, where $i$ and $j$ are histogram
bins.} The histograms themselves were constructed by posterizing random images
from an online image repository~\citep{flickr}. For the first two data sets,
the data set consisted of $100\,000$ vectors, with $20$ and $50$ pivots,
respectively. For the last data set, the~100 queries and~20 pivots were
sampled without replacement from the original set of $10\,000$ images, and
removed before the searches were performed.
The pivot filtering was performed directly on a precomputed distance table
in the style of \textsc{laesa}~\citep{Mico:1994}. In addition to the full
Ptolemaic pivot filtering described in \sect~\ref{sec:pivots}, a limited
version was tested with only the $n-1$ consecutive pairs of pivots, in an
arbitrary ordering, giving $2n-2$ bounds of the form of (\ref{eq:pt}), because
the absolute value was used. This gives a \textsc{cpu} use closer to that of
traditional pivoting. Of course, there is a range of settings available here,
from a linear number of bounds, to the quadratic number used by the full
filter.
The plots in \fig~\ref{fig:filtering} show the total number distance
computations needed beyond the number of pivots (i.e., the number of objects
not eliminated).
The search radius (the horizontal axis) is described using the
number of objects it encompasses beyond the query (10--50). In each case, $m$
is the number of
pivots used. The results were averaged over~100 random queries.
As can be seen, for these cases the full Ptolemaic approach clearly
outperforms the triangular, with the partial Ptolemaic filtering
also offering an advantage of varying significance.

\pgfplotstableread{dcmp_unif_10d_k01-10.dat}\dcmpa
\pgfplotstableread{dcmp_unif_10d_k10-50.dat}\dcmpb
\pgfplotstableread{dcmp_clus_20d_k01-10.dat}\dcmpe
\pgfplotstableread{dcmp_clus_20d_k10-50.dat}\dcmpf
\pgfplotstableread{dcmp_imgs_k01-10.dat}\dcmpg
\pgfplotstableread{dcmp_imgs_k10-50.dat}\dcmph

\begin{figure}
\begin{center}
    \subfigure[$L_2$, uniform, $\R^{10}$, $m=20$, $n=100\,000$]{\begin{tikzpicture}
    \begin{axis}[
        small,
        legend pos=north west,
        scaled y ticks=base 10:-4,
        cycle list name=black white,
        legend style={
            at={(0.03,0.96)},
            anchor=north west,
            font=\footnotesize,
            draw=none,
            fill=none
            },
        legend cell align=left,
        ]
        \addplot
            table[mark=square,y=met,x=k] {\dcmpa};
        \addplot+[dashed,every mark/.append style={solid,fill=white}]
            table[y=cpt,x=k] {\dcmpa};
        \addplot+[mark=square*]
            table[y=pto,x=k] {\dcmpa};
        \legend{Triangular,Partial~Ptol.,Ptolemaic}
\end{axis}
\end{tikzpicture}
\qquad
\begin{tikzpicture}
    \begin{axis}[
        small,
        legend pos=north west,
        scaled y ticks=base 10:-4,
        cycle list name=black white,
        ]
        \addplot
            table[mark=square,y=met,x=k] {\dcmpb};
        \addplot+[dashed,every mark/.append style={solid,fill=white}]
            table[y=cpt,x=k] {\dcmpb};
        \addplot+[mark=square*]
            table[y=pto,x=k] {\dcmpb};
\end{axis}
\end{tikzpicture}\label{fig:piv:a}}\\
\subfigure[$L_2$, clustered, $\R^{20}$, $m=50$, $n=100\,000$]{\begin{tikzpicture}
    \begin{axis}[
        small,
        legend pos=south east,
        scaled y ticks=base 10:-4,
        cycle list name=black white,
        ]
        \addplot
            table[mark=square,y=met,x=k] {\dcmpe};
        \addplot+[dashed,every mark/.append style={solid,fill=white}]
            table[y=cpt,x=k] {\dcmpe};
        \addplot+[mark=square*]
            table[y=pto,x=k] {\dcmpe};
        \addplot[draw=none] coordinates {(1,0)};
\end{axis}
\end{tikzpicture}%
\qquad
    \begin{tikzpicture}
    \begin{axis}[
        small,
        legend pos=north west,
        scaled y ticks=base 10:-4,
        cycle list name=black white,
        ]
        \addplot
            table[mark=square,y=met,x=k] {\dcmpf};
        \addplot+[dashed,every mark/.append style={solid,fill=white}]
            table[y=cpt,x=k] {\dcmpf};
        \addplot+[mark=square*]
            table[y=pto,x=k] {\dcmpf};
        \addplot[draw=none] coordinates {(10,0)};
\end{axis}
\end{tikzpicture}\label{fig:piv:b}}\\
\subfigure[QFD, image histograms, $\R^{64}$, $m=20$, $n=9880$]{\begin{tikzpicture}
    \begin{axis}[
        small,
        legend pos=south east,
        scaled y ticks=base 10:-3,
        cycle list name=black white,
        ]
        \addplot
            table[mark=square,y=met,x=k] {\dcmpg};
        \addplot+[dashed,every mark/.append style={solid,fill=white}]
            table[y=cpt,x=k] {\dcmpg};
        \addplot+[mark=square*]
            table[y=pto,x=k] {\dcmpg};
        \addplot[draw=none] coordinates {(1,0)};
\end{axis}
\end{tikzpicture}%
\qquad
    \begin{tikzpicture}
    \begin{axis}[
        small,
        legend pos=north west,
        scaled y ticks=base 10:-3,
        cycle list name=black white,
        ]
        \addplot
            table[mark=square,y=met,x=k] {\dcmph};
        \addplot+[dashed,every mark/.append style={solid,fill=white}]
            table[y=cpt,x=k] {\dcmph};
        \addplot+[mark=square*]
            table[y=pto,x=k] {\dcmph};
        \addplot[draw=none] coordinates {(10,0)};
\end{axis}
\end{tikzpicture}\label{fig:piv:c}}

\caption{
Distance computations ($y$-axis) for range search with $m$ pivots over $n$
objects, radii covering 1--50 neighbors ($x$-axis), averaged over 100 random
queries. The clustered data consisted of~10 Gaussian clusters, each with
$\sigma^2=0.1$. The queries for the uniform case were uniform. For the
clusters,~10 queries were generated from each cluster distribution; for the
image histograms, the pivots and queries were sampled without replacement, and
removed from the original~10\,000 vectors.}\label{fig:filtering}
\end{center}
\end{figure}
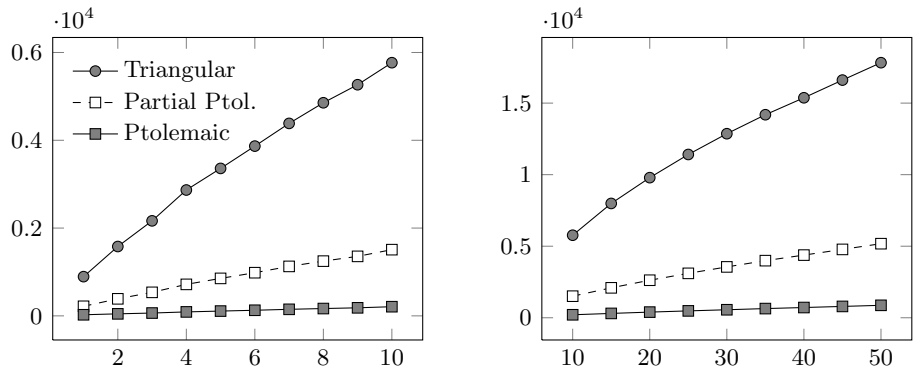
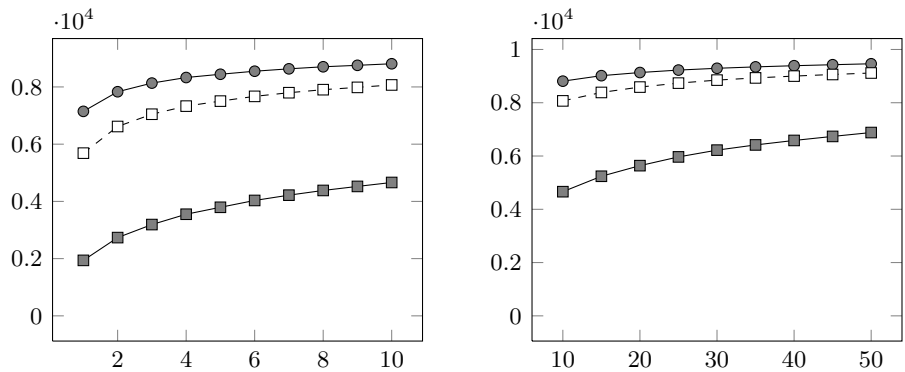
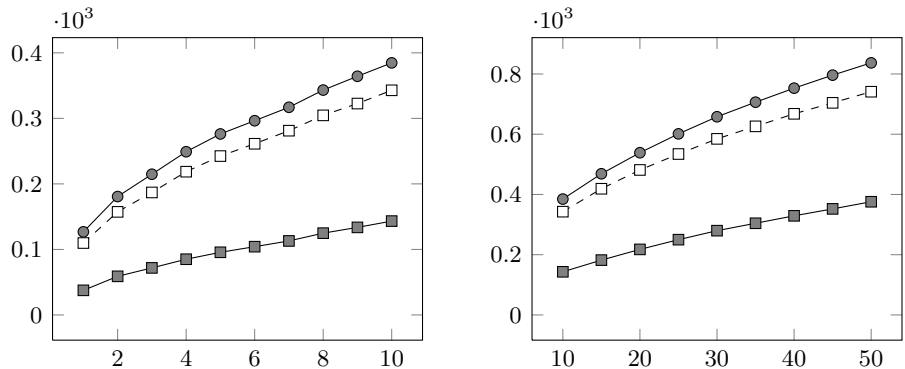

It is worth noting that the Ptolemaic filters used in
\fig~\ref{fig:filtering} are \emph{pure}, that is, not combined with a
triangular bound. \fig~\ref{fig:filtercontrib} illustrates the relative
contribution of the triangular and Ptolemaic pivoting bounds, if both are
used in range search where the radius covers the ten nearest objects, on
$100\,000$ uniformly random ten-dimensional vectors, averaged over $100$
queries. Queries and pivots were drawn from the same distribution as the data.
The filtering power is measured as the number of objects eliminated by the
given bounds during search.
The four regions represent
\begin{inparaenum}[(\itshape 1\upshape)]
\item objects filtered out only by the Ptolemaic bound,
\item objects that both bounds are able to filter out,
\item objects filtered out only by the triangular bound, and
\item objects that neither bound manages to disqualify.
\end{inparaenum}
A progression similar to that in \fig~\ref{fig:filtercontrib} was also found
for clustered data, with a somewhat less pronounced difference (data not
shown).

\pgfplotstableread{filpow_10k_10d_2-10p.dat}\pfa
\pgfplotstableread{filpow_10k_10d_10-50p.dat}\pfb

\pgfplotstableread{filpow_10k_20d_10-50p.dat}\pfc

\begin{figure}
\begin{center}
    \subfigure{\begin{tikzpicture}
    \begin{axis}[
        xmin=0,
        small, xbar stacked,
        ylabel=Pivots,
        xlabel={Filtered, $\R^{10}$}
        ]
        \addplot[fill=darkgray]
            table[x=pto,y=piv] {\pfa};
        \addplot[fill=gray]
            table[x=both,y=piv] {\pfa};
        \addplot[fill=lightgray]
            table[x=met,y=piv] {\pfa};
        \addplot[fill=white]
            table[x=none,y=piv] {\pfa};
\end{axis}
\end{tikzpicture}}
\qquad
\subfigure{\begin{tikzpicture}
    \begin{axis}[
        xmin=0,
        small, xbar stacked,
        xlabel={Filtered, $\R^{20}$},
        area legend,
        legend style={
        font=\footnotesize,
        at={(0.03,0.96)},
        anchor=north west,
             draw=none,
        },
        legend cell align=left,
        ]
        \addplot[fill=darkgray]
            table[x=pto,y=piv] {\pfc};
        \addplot[fill=gray]
            table[x=both,y=piv] {\pfc};
        \addplot[fill=lightgray]
            table[x=met,y=piv] {\pfc};
        \addplot[fill=white]
            table[x=none,y=piv] {\pfc};
        \legend{Ptolemaic,Both,Triangular,Neither}
\end{axis}
\end{tikzpicture}}
\caption{Filtering contributions on uniformly random vectors, for the given
number of random pivots. The shaded areas represent the number of objects that
were filtered exclusively by the triangular and Ptolemaic bounds, by both, or
by neither.}\label{fig:filtercontrib}
\end{center}
\end{figure}
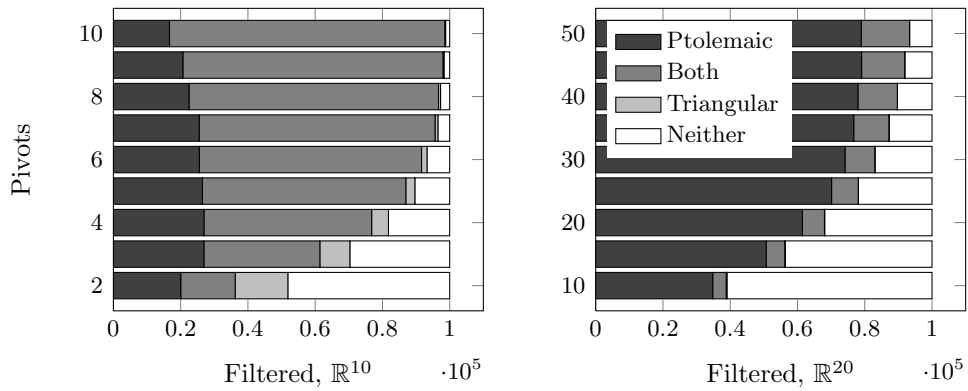

\begin{figure}
\begin{center}
\subfigure{\begin{tikzpicture}
    \begin{axis}[
        ylabel=Filtered,
        xlabel=Pivots,
        cycle list name=black white,
        legend style={
            at={(0.03,0.96)},
            anchor=north west,
            font=\footnotesize,
            draw=none,
            fill=none
            },
        legend cell align=left,
        ]
        \addplot
            table[mark=square,y expr=\thisrow{met}+\thisrow{both},x=piv] {\pfc};
        \addplot+[dashed,every mark/.append style={solid,fill=white}]
            table[y=cpt,x=piv] {\pfc};
        \addplot+[mark=square*]
            table[y expr=\thisrow{pto}+\thisrow{both},x=piv] {\pfc};
        \addplot[draw=none] coordinates {(10,0) (50,100000)};
        \legend{Triangular,Partial Ptolemaic,Ptolemaic}
\end{axis}
\end{tikzpicture}}
\caption{Individual filtering power for uniformly random vectors
from $\R^{20}$. The $y$-axis is the number of objects filtered by each of the
three bounds individually, independently of the other two, for the given number
of random pivots.} \label{fig:filterseparate}
\end{center}
\end{figure}
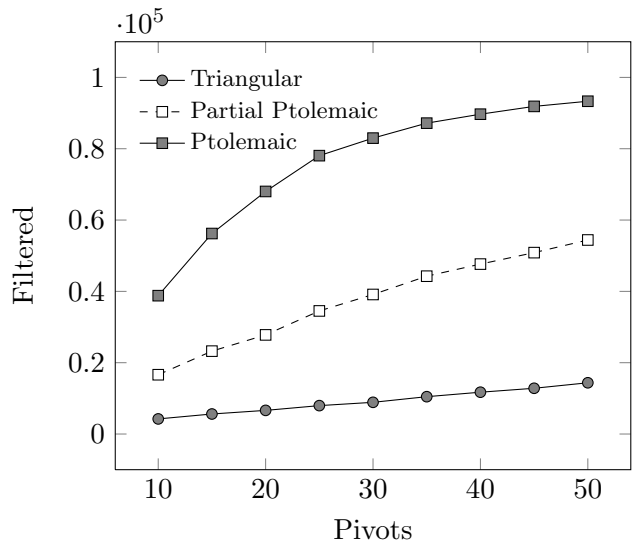

Even for a relatively modest number of pivots, the filtering power of the
Ptolemaic bound is high, and the difference in the exclusive contributions of
the two bounds is surprisingly large, growing to the extreme for higher
dimensions.
This can perhaps be seen even more clearly in \fig~\ref{fig:filterseparate},
where the individual (non-exclusive) filtering powers of the triangular,
partial Ptolemaic and full Ptolemaic methods are plotted independently.

\subsection{Approximation Rates for Non-Ptolemaic Metrics}\label{sec:approx}

Although I focus on exact search in this paper, the ideas may be applicable
to approximate search as well, and could perhaps also be used with merely
\emph{approximately} Ptolemaic distances.
While any metric
can be made Ptolemaic, as mentioned in \sect~\ref{sec:basic}, the resulting
metric may be harder to index, and approximate indexing of the original may
still be useful. Table~\ref{tab:approx} shows results for several data sets
and distances. The three types of objects used were vectors, sets and strings.
The vectors were generated as described in \sect~\ref{sec:filtering}. The sets
were generated randomly for a given maximum cardinality by including or
excluding each object with equal probability. The string data sets were a list
of $234\,936$ words from Webster's Second International Edition, as
distributed with the Macintosh OS version 10.5, and
the lines of \emph{A Tale of Two Cities} by Charles Dickens,
with short lines (fewer than six characters) stripped away. The latter data
set is similar to one used by \mbox{\citet{Brin:1995}}.
For the vector data sets, various $L_p$ norms were used, as well as the
angular pseudometric. For the sets, the related Hamming and Jaccard distances
were used (the cardinality of the symmetric difference, and proportion of the
symmetric difference to the union, respectively), while for the strings,
Levenshtein distance (edit distance) was used.
Each experiment was averaged over ten runs. Except for the string data sets,
which were static, new data sets were randomly generated for each run. A run
consisted of $10\,000$ randomly sampled quadruples (that is, sets of four
distinct objects), and the proportion of the quadruples that satisfied
Ptolemy's inequality was computed.

As can be seen from Table~\ref{tab:approx}, the much-used $L_p$ norms vary in
their approximation rates from 99\%
(with $L_2$, of course having an exact rate of 100\%
pseudometric does very well for the high-dimensional vectors.
The set distances also seem to conform to Ptolemy's inequality to a high
degree for the 20-dimensional case, and the edit distance has a very low rate
of inequality violations (with the Dickens data set outperforming the Webster
data set by several orders of magnitude).
For the $L_p$ spaces and the angular pseudometric, the intrinsic
dimensionality
was generally close to the actual number of dimensions (slightly higher for
uniformly random vectors, and varying with $p$).
The dimensionalities for the Hamming spaces were equal to about half the set
cardinalities, while those for Jaccard distance were about 50\%
intrinsic dimensionality of the Webster data set was about nine, while that
for the Dickens data set was about forty.
\newcommand{\T}{\phantom{\E{-1}}}
\begin{table}[!t]
\begin{threeparttable}
\caption{Results from approximation experiments. The values are
average Ptolemaic proportions of quadruples, as well as standard
deviations. Subscripts on numbers represent powers of ten. $\R^k$ refers to
vectors of dimension $k$, while $2^{\Z_{k}}$ are sets of cardinality
$k$.}\label{tab:approx}
\begin{tabularx}{\linewidth}{Xcccccc}
\toprule
\bf Distance & \multicolumn{2}{c}{Clustered, $\R^5$}
             & \multicolumn{2}{c}{Clustered, $\R^{10}$}
             & \multicolumn{2}{c}{Uniform,   $\R^{10}$} \\
             \cmidrule(r){2-3} \cmidrule(r){4-5} \cmidrule(r){6-7}
              & $\mu$ & $\sigma$ & $\mu$ & $\sigma$ & $\mu$ & $\sigma$ \\
\midrule
\makebox[1cm][l]{$L_1$}Manhattan distance
             & 0.98 & 1.70\E{-3} & 1.00 & 3.81\E{-4} & 1.00 & 8.63\E{-4} \\
\makebox[1cm][l]{$L_2$}Euclidean distance\tnote{$*$}
             & 1.00 & ---        & 1.00 & ---        & 1.00 & ---        \\
$L_3$        & 0.99 & 6.39\E{-4} & 1.00 & 1.11\E{-4} & 1.00 & 4.58\E{-5} \\
$L_5$        & 0.98 & 1.41\E{-3} & 1.00 & 6.03\E{-4} & 1.00 & 2.14\E{-4} \\
$L_{10}$     & 0.97 & 1.39\E{-3} & 0.99 & 1.31\E{-3} & 1.00 & 4.94\E{-4} \\
$L_{100}$    & 0.96 & 2.01\E{-3} & 0.98 & 1.49\E{-3} & 1.00 & 5.27\E{-4} \\
\makebox[1cm][l]{$L_\infty$}Chebyshev distance
             & 0.96 & 2.09\E{-3} & 0.98 & 1.45\E{-3} & 1.00 & 5.48\E{-4} \\
\makebox[1cm][l]{$\theta$}Angular distance & 0.99
             & 5.78\E{-4} & 1.00 & 1.19\E{-4} & 1.00 & 4.00\E{-5}\\
\midrule
                   & \multicolumn{2}{c}{Uniform, $2^{\Z_{10}}$}
                     & \multicolumn{2}{c}{Uniform, $2^{\Z_{20}}$} \\
                    \cmidrule(r){2-3} \cmidrule(r){4-5}
Hamming distance                 & 0.93 & 2.69\E{-3} & 0.98 & 9.57\E{-4} \\
Jaccard distance                 & 0.99 & 8.90\E{-4} & 1.00 & 1.60\E{-4} \\
\midrule
                   & \multicolumn{2}{c}{Webster}
                     & \multicolumn{2}{c}{Dickens} \\
                    \cmidrule(r){2-3} \cmidrule(r){4-5}
Levenshtein distance             & 1.00 & 7.81\E{-5} & 1.00 & 0.00 \\
\bottomrule
\end{tabularx}
\begin{tablenotes}
\footnotesize\item[$*$] Euclidean distance is Ptolemaic, but is included for
completeness
\end{tablenotes}
\end{threeparttable}
\end{table}

\section{Discussion and Future Work}

In summary, this paper has presented a new distance indexing principle, based
on Ptolemy's inequality, which seems to result in greatly increased filtering
power, in the cases where it is applicable. For vector spaces (in particular,
$\R^k$), it is directly applicable to quadratic form distances, including
Euclidean distance. The cost of the filtering depends on the number of bounds
used for a given number of pivots, growing from linear, as with triangular
filtering, up to quadratic, providing a spectrum of options which could form
the basis of a tradeoff between filtering cost and the performance gained by
eliminating distance computations. The exact tradeoff will of course depend on
the computational cost of the distance function. The quadratic form distance
is, at the face of it, an excellent candidate in this regard, as it is quite
expensive to compute. However, static quadratic form distances can be
normalized to Euclidean distance, transforming the query vector prior to
search. This leaves dynamic versions, such as the signature quadratic form
distance, as an important application. Results published after the original
version of this paper shows that Ptolemaic indexing is, indeed, currently the
most efficient way to index these distances~\citep{Lokoc:2011,Hetland:2013}.
Beyond such empirical validation in a more realistic setting, there are other
avenues of research that might be interesting to pursue. For example, as
shown, there also several common distances that are \emph{close} to being
Ptolemaic. Whether Ptolemaic indexing of such spaces might turn out to be
useful as an approximate indexing method could be an issue worth investigating
further.
It might also be interesting to examine the basic properties of the Ptolemaic
pivots, as related to the space (including which composition of pivots that
would give the best pivoting results). Some basic engineering work in this
direction has already met with experimental success, such as exploring pivot
pairs in order of distance from the query and the object in question, to
achieve elimination more quickly~\citep{Lokoc:2011,Hetland:2013}.

It seems that higher dimensionalities give Ptolemaic indexing a greater edge
over the triangular one. In the experiments performed here, this may be due to
the increased number of pivots, but even the partial Ptolemaic filtering
increases its lead with with higher dimension and a higher pivot number, which
cannot be explained by the quadratic growth in the number of bounds. This is
certainly an issue worthy of further examination. It also seems like
high-dimensional spaces are more likely to be approximately
Ptolemaic---perhaps because the distances are more similar. Based on this
reasoning, high-dimensional spaces would also be more likely to be
approximately triangular. Taking the square root, which will
increase the intrinsic dimension, will make any metric \emph{exactly}
Ptolemaic. On the other hand, both the Ptolemaic and the triangular bound
will be weaker in these cases, so there is clearly a tradeoff between accuracy
and filtering power. These are issues that could be examined both empirically
and analytically.

\appendix

\section*{Acknowledgements}

The author wishes to thank Ole Edsberg, Jon Marius Venstad and Bilegsaikhan
Naidan for their help in several areas in the preparation of this paper, as
well as the anonymous reviewers for their helpful suggestions.

\bibliography{paper,local}

\end{document}